\documentclass{article}

\usepackage{amsmath,amsfonts,amsthm,braket,graphicx,subfig}

\theoremstyle{definition} 
\newtheorem {theorem} {Theorem}

\title{History Dependent Quantum Walk on the Cycle with an Unbalanced Coin}
\author{Walter O. Krawec\\\small{Stevens Institute of Technology}\\\small{Hoboken NJ, 07030 USA}\\\small{\texttt{walter.krawec@gmail.com}}}

\newcommand{\MoveLeft}{\downarrow}
\newcommand{\MoveRight}{\uparrow}

\begin{document}
\maketitle

\begin{abstract}
Recently, a new model of quantum walk, utilizing recycled coins, was introduced; however little is yet known about its properties.  In this paper, we study its behavior on the cycle graph.  In particular, we will consider its time averaged distribution and how it is affected by the walk's ``memory parameter'' - a real parameter, between zero and eight, which affects the walk's coin flip operator.  Despite an infinite number of different parameters, our analysis provides evidence that only a few produce non-uniform behavior.  Our analysis also shows that the initial state, and cycle size modulo four all affect the behavior of this walk.  We also prove an interesting relationship between the recycled coin model and a different memory-based quantum walk recently proposed.
\end{abstract}


\section{INTRODUCTION}

Quantum random walks, the quantum analogue of classical random walks, have been a topic of much interest lately due in large part to their many properties and applications including (just to list a few): searching \cite{QW-search1}, subset finding \cite{QW-subset-finding}, and they have been shown to be universal for quantum computation \cite{QW-universal1,QW-universal2}.  The reader is referred to Ref \cite{QW-intro1} for an introduction to quantum walks, and Ref \cite{QW-survey} for a general survey.

Recently, several new models of history dependent quantum walks have been proposed including the quantum walk with memory \cite{QW-memory1d,QW-memory-cycle}, the non-repeating quantum walk \cite{QW-non-repeating}, and the quantum walk with recycled coins \cite{QW-recy-coin}.  However, due to their relatively young age and their increased complexity, they are not as well analyzed as their non-history dependent counterparts.  It is the recycled coin model that is the main topic of interest in this paper.

The recycled coin model utilizes $N$ coin spaces (we will restrict our attention to $N=2$ in this paper) and a memory parameter which is a continuous variable, $\phi \in [0,8)$, that dictates how the coin flip operator, which acts on both coin spaces, should behave.  Despite this multitude of possible settings, in this paper we observe that only a few memory parameter values seem to actually induce different limiting distributions; all others appear to induce the uniform distribution.  We will observe that the initial state, the cycle size modulo four, and whether the cycle size is divisible by $12$ or not, all affect the limiting distribution of the walk.  Also, we will prove certain memory parameters produce equivalent limiting distributions depending on the initial state.  Finally, we will prove an interesting connection between the recycled coin model and the memory model of Ref \cite{QW-memory-cycle}.  We will prove some of our observations, and verify others numerically, leaving their proofs as very interesting open questions.

\section{QUANTUM WALKS}

A discrete time quantum walk on the cycle of size $d \in \mathbb{N}$, (see Ref \cite{QW-cycle-first}), operates over the Hilbert space $\mathcal{H}_P \otimes \mathcal{H}_C$, where $\mathcal{H}_P$ is the position space, spanned by orthonormal basis $\{\ket{i} \text{ } | \text{ } i = 0, 1, \cdots, d-1\}$ and $\mathcal{H}_C$ is the coin space spanned by orthonormal basis $\{\ket{\MoveLeft}, \ket{\MoveRight}\}$.  At each time step of the walk, the walker, which begins in some initial state $\ket{\psi_0} \in \mathcal{H}_P \otimes \mathcal{H}_C$ (for example $\ket{\psi_0} = \ket{0,\MoveLeft}$), undergoes the unitary evolution $U = S \cdot (I_P \otimes C)$.  Here, $C$ is referred to as the \emph{coin flip operator} and it acts only on $\mathcal{H}_C$ (for example $C$ may be the Hadamard operator); $I_P$ is the identity operator on $\mathcal{H}_P$; and $S$ is the \emph{shift operator} which updates the walker's position based on the state of its coin space:
\begin{equation}\label{eq:shift}
S = \sum_{n=0}^{d-1} \ket{n+1 \text{ mod } d}\bra{n} \otimes \ket{\MoveRight}\bra{\MoveRight} + \sum_{n=0}^{d-1} \ket{n-1 \text{ mod } d}\bra{n} \otimes \ket{\MoveLeft}\bra{\MoveLeft}.
\end{equation}
That is, $S$ maps basis states $\ket{n,\MoveLeft} \mapsto \ket{n-1,\MoveLeft}$ and $\ket{n, \MoveRight} \mapsto \ket{n+1, \MoveRight}$ with all arithmetic done modulo the cycle size $d$.  After $t$ time steps, the walker is in the state $U^t\ket{\psi_0}$.

\subsection{With Memory}

In \cite{QW-memory-cycle}, an extension of this model, based on the memory model of \cite{QW-memory1d} was proposed: the discrete time quantum walk with memory on the cycle.  Now, the walker lives in the Hilbert space $\mathcal{H}_P \otimes \mathcal{H}_M \otimes \mathcal{H}_C$; exactly the same space as the standard model above, but with the additional memory ancilla $\mathcal{H}_M$.  This new space is spanned by the orthonormal basis $\{\ket{\MoveLeft}, \ket{\MoveRight}\}$ and it is used to ``keep track'' of the previous position of the particle one time step in the past.

In greater detail, this walk evolves at each time step via the unitary operator $U_{\mathcal{M}} = S_{\mathcal{M}} \cdot (I_P \otimes I_M \otimes C)$ where $C$ is the coin flip operator acting only on $\mathcal{H}_C$ (as before); $I_P$ and $I_M$ are the identity operators on $\mathcal{H}_P$ and $\mathcal{H}_M$ respectively; and $S_{\mathcal{M}}$ is the shift operator which acts on basis states as follows:
\[
\begin{array}{c|l}
\ket{n,\MoveLeft,\MoveLeft} \mapsto \ket{n-1,\MoveLeft,\MoveLeft} & \ket{n,\MoveLeft, \MoveRight} \mapsto \ket{n+1, \MoveRight,\MoveRight}\\
\ket{n,\MoveRight,\MoveLeft} \mapsto \ket{n+1, \MoveRight, \MoveLeft} & \ket{n, \MoveRight, \MoveRight} \mapsto \ket{n-1, \MoveLeft, \MoveRight}.
\end{array}
\]

We will return to this model later.

\subsection{With Recycled Coins}

We now review the recycled coin model described in \cite{QW-recy-coin} (an extension of \cite{QW-history,QW-many-coins}), which is our primary interest in this paper.  This system operates over the Hilbert space:

\begin{equation}
\mathcal{H} = \mathcal{H}_P \otimes \mathcal{H}_{C_1} \otimes \cdots \otimes \mathcal{H}_{C_N},
\end{equation}
where $N \ge 1$, $\mathcal{H}_{C_i}$ is the $i$'th coin space, and $\mathcal{H}_P$ is the position space.  In this paper, we are interested in studying the behavior of this walk on the cycle of size $d$.  We also restrict our attention to the case when there are only two coin spaces (i.e., $N=2$); note that even in this ``restricted'' scenario, little is known of this walk's behavior and, to our knowledge, it has not been studied on the cycle.  Also, by limiting our investigation to $N=2$, we are able to compare this walk with results recently found in \cite{QW-memory-cycle}, which discussed the behavior of the quantum walk with memory on the cycle, discussed last section.

In this paper, each $\mathcal{H}_{C_i}$ ($i=1,2$) is spanned by the orthonormal basis $\{\ket{\MoveLeft}_i, \ket{\MoveRight}_i\}$ (when the context is clear, we will forgo writing the subscript ``$i$''), while the position space $\mathcal{H}_P$ is spanned by the orthonormal basis $\{\ket{i} \text{ } | \text{ } i = 0, 1, \cdots, d-1\}$.

Initially, the walker begins in some state $\ket{\psi_0} \in \mathcal{H}$, for instance $\ket{\psi_0} = \ket{0}\ket{\MoveLeft, \MoveLeft}$.  At each time step, a coin flip operator is applied which affects the state of the \emph{active coin}: $\mathcal{H}_{C_2}$ (the ``active coin,'' which is always $\mathcal{H}_{C_N}$, is the coin which determines the behavior of the shift operator described next) but which acts on both $\mathcal{H}_{C_1}\otimes\mathcal{H}_{C_2}$ (this distinguishes this model from the multi-coin walk described in Ref \cite{QW-many-coins}).  Let:
\begin{equation}
C(\theta)  = \left(\begin{array}{cc}
\cos\theta & \sin\theta\\
\sin\theta & -\cos\theta
\end{array}
\right),
\end{equation}
be the operator which maps $\ket{\MoveLeft} \mapsto \cos\theta\ket{\MoveLeft} + \sin\theta\ket{\MoveRight}$ and $\ket{\MoveRight} \mapsto \sin\theta\ket{\MoveLeft} - \cos\theta\ket{\MoveRight}$ (when $\theta = \pi/4$ this is the Hadamard operator).  We define the coin flip operator to be:
\begin{equation}\label{eq:coin-flip}
\widehat{C} = \ket{\MoveLeft}_1\bra{\MoveLeft}_1 \otimes C\left(\frac{\pi}{4}\right) + \ket{\MoveRight}_1\bra{\MoveRight}_1 \otimes C\left(\frac{\pi}{4}(1+\phi)\right).
\end{equation}

This coin is a slightly modified version of a more balanced one described in \cite{QW-recy-coin}; we call it an \emph{unbalanced coin} and it leads to some very interesting behavior.  Note that, as mentioned in \cite{QW-recy-coin}, there are many choices of coin flip operator and memory function.  Here $\phi$ is the \emph{memory parameter}; observe that, if $\phi = 0$, then the state of $\mathcal{H}_{C_1}$ has no influence on $\widehat{C}$'s action on $\mathcal{H}_{C_2}$.  Indeed, in this case, the recycled coin model becomes equivalent to the multi-coin walk described in \cite{QW-many-coins}.  For $\phi \ne 0$, however, the action on $\mathcal{H}_{C_2}$ will depend on the state of $\mathcal{H}_{C_1}$.  Due to the periodicity of $C(\theta)$, we need only consider $\phi \in [0,8)$.

Following the coin flip operator, a shift operator is applied.  This operator behaves exactly as in the standard quantum walk model (see Equation \ref{eq:shift}); however the action on $\mathcal{H}_P$ depends only on the state of the $\mathcal{H}_{C_2}$ system; the state of $\mathcal{H}_{C_1}$ does not matter.  That is to say, $\mathcal{H}_{C_2}$, being the active coin, determines the direction the particle moves, the state of $\mathcal{H}_{C_1}$ is irrelevant at this point.  This shift operator will map basis states $\ket{n, c_1, \MoveLeft}$ to $\ket{n-1, c_1, \MoveLeft}$ and $\ket{n, c_1, \MoveRight}$ to $\ket{n+1, c_1, \MoveRight}$ where all arithmetic is done modulo $d$ and $c_1 \in \{\MoveLeft, \MoveRight\}$.

Finally, a \emph{memory update} operator is applied which simply swaps the active coin for the inactive one.  That is:
\begin{equation}\label{eq:memory-update}
M = \ket{\MoveLeft \MoveLeft}\bra{\MoveLeft \MoveLeft} + \ket{\MoveRight \MoveLeft}\bra{\MoveLeft \MoveRight} + \ket{\MoveLeft \MoveRight}\bra{\MoveRight\MoveLeft} + \ket{\MoveRight \MoveRight}\bra{\MoveRight \MoveRight}.
\end{equation}

With these operators defined, the walk evolves over time via the unitary operator:
\begin{equation}\label{eq:evolve}
U = (I_P\otimes M) \cdot S \cdot (I_P \otimes \widehat{C}),
\end{equation}
where $I_P$ is the identity operator on $\mathcal{H}_P$.  If the initial state is $\ket{\psi_0}$, then, after $t$ steps of the walk, the state evolves to $\ket{\psi_t} = U^t\ket{\psi_0}$.

\section{ANALYSIS ON THE CYCLE}

We are interested in two quantities.  First, the probability that, after $t$ steps, the walker lands on position $\ket{n} \in \mathcal{H}_P$, provided a measurement of $\mathcal{H}_P$ were performed.  This value is:

\begin{equation}\label{eq:pr}
p(n,t, \phi ; \psi_0) = \sum_{c_1,c_2} \left|\braket{n, c_1, c_2 | U^t|\psi_0}\right|^2,
\end{equation}
where the sum is over $c_1,c_2 \in \{\MoveLeft, \MoveRight\}$, and $\phi$ is the memory parameter (on which $U$ depends).  Due to the nature of the quantum walk, this function does not converge as $t \rightarrow \infty$ \cite{QW-cycle-first}; thus, we consider also the \emph{time averaged distribution} defined as:

\begin{equation}\label{eq:time-avg}
\bar{p}(n, \phi ; \psi_0) = \lim_{T \rightarrow \infty} \frac{1}{T} \sum_{t=1}^T p(n,t,\phi; \psi_0),
\end{equation}
which does converge \cite{QW-cycle-first}.  This distribution is also called the \emph{limiting distribution}.

Note that, when the initial state is clear, or is irrelevant, we will often forgo writing the ``$;\psi_0$'' portion of the above functions.

We now investigate the recycled coin model on the $d$-cycle.  In particular, we are interested in understanding the behavior of $\bar{p}(n, \phi)$ for various memory parameters.  To do so, we will use Fourier analysis.  This is a technique commonly used to analyze quantum walks \cite{QW-survey,QW-Fourier}; in particular, we will, at first, primarily be following the example of Ref \cite{QW-memory-cycle}, which used this method to analyze the properties of the quantum walk with memory on the cycle.

\subsection{Fourier Analysis}

At time $t$, we may write the state of the walk $\ket{\psi_t}$ as:
\[
\ket{\psi_t} = \sum_{n=0}^{d-1}\ket{n}\otimes \left[ \sum_{c_1,c_2}\psi_{c_1,c_2}(n,t)\ket{c_1,c_2} \right] = \sum_{n=0}^{d-1}\ket{n}\otimes\psi(n,t),
\]
where the sum is over all $c_1,c_2 \in \{\MoveLeft,\MoveRight\}$, $\psi_{c_1,c_2}(n,t)$ is the probability amplitude corresponding to the basis ket $\ket{n,c_1,c_2}$ at time $t$, and:
\[
\psi(n,t) = \left(\begin{array}{c}
\psi_{\MoveLeft\MoveLeft}(n,t)\\
\psi_{\MoveLeft\MoveRight}(n,t)\\
\psi_{\MoveRight\MoveLeft}(n,t)\\
\psi_{\MoveRight\MoveRight}(n,t)
\end{array}
\right).
\]

It is clear that $p(n,t,\phi) = tr\left(\psi^*(n,t)\psi(n,t)\right)$ (where $\psi^*(n,t)$ represents the conjugate transpose of $\psi(n,t)$, $tr(\cdot)$ is the trace operation, and $p(n,t)$ is from Equation \ref{eq:pr}).

We first compute $\psi(n, t+1)$.  Consider the effect of applying the coin flip operator $(I_P\otimes\widehat{C})$ to the state $\ket{n}\otimes\psi(n,t)$:

\[
(I_P \otimes \widehat{C})\ket{n}\otimes\psi(n,t) = \ket{n}\otimes\left(\begin{array}{c}
\frac{1}{\sqrt{2}}\psi_{\MoveLeft\MoveLeft}(n,t) + \frac{1}{\sqrt{2}}\psi_{\MoveLeft\MoveRight}(n,t)\\
\frac{1}{\sqrt{2}}\psi_{\MoveLeft\MoveLeft}(n,t) - \frac{1}{\sqrt{2}}\psi_{\MoveLeft\MoveRight}(n,t)\\
\cos\theta\psi_{\MoveRight\MoveLeft}(n,t) + \sin\theta\psi_{\MoveRight\MoveRight}(n,t)\\
\sin\theta\psi_{\MoveRight\MoveLeft}(n,t) - \cos\theta\psi_{\MoveRight\MoveRight}(n,t)
\end{array}\right),
\]
where $\theta = \frac{\pi}{4}(1+\phi)$.

Of course, the probability amplitudes at position $\ket{n}$, after applying the shift operator $S$, are determined only by the amplitudes at positions $\ket{n-1}$ and $\ket{n+1}$ (addition modulo $d$) before the shift operator was applied.  From the above equation, and recalling that the shift operator's action is determined only by the active coin - the right most in our case - it is evident that this value is:

\begin{align*}
&S\cdot(I_P\otimes\widehat{C})\left(\sum_{n=0}^{d-1}\ket{n}\otimes\psi(n,t)\right)\\
&= \ket{n}\otimes\left(\begin{array}{c}
\frac{1}{\sqrt{2}}\psi_{\MoveLeft\MoveLeft}(n+1,t) + \frac{1}{\sqrt{2}}\psi_{\MoveLeft\MoveRight}(n+1,t)\\
\frac{1}{\sqrt{2}}\psi_{\MoveLeft\MoveLeft}(n-1,t) - \frac{1}{\sqrt{2}}\psi_{\MoveLeft\MoveRight}(n-1,t)\\
\cos\theta\psi_{\MoveRight\MoveLeft}(n+1,t) + \sin\theta\psi_{\MoveRight\MoveRight}(n+1,t)\\
\sin\theta\psi_{\MoveRight\MoveLeft}(n-1,t) - \cos\theta\psi_{\MoveRight\MoveRight}(n-1,t)
\end{array}\right) + \cdots
\end{align*}
where above, we have written only the relevant $\ket{n}$ potion of the resulting state.

Finally, the memory update operator $M$ is applied yielding:
\begin{align*}
&(I_P\otimes M)\cdot S\cdot(I_P\otimes\widehat{C})\left(\sum_{n=0}^{d-1}\ket{n}\otimes\psi(n,t)\right)\\
&= \ket{n}\otimes\left(\begin{array}{c}
\frac{1}{\sqrt{2}}\psi_{\MoveLeft\MoveLeft}(n+1,t) + \frac{1}{\sqrt{2}}\psi_{\MoveLeft\MoveRight}(n+1,t)\\
\cos\theta\psi_{\MoveRight\MoveLeft}(n+1,t) + \sin\theta\psi_{\MoveRight\MoveRight}(n+1,t)\\
\frac{1}{\sqrt{2}}\psi_{\MoveLeft\MoveLeft}(n-1,t) - \frac{1}{\sqrt{2}}\psi_{\MoveLeft\MoveRight}(n-1,t)\\
\sin\theta\psi_{\MoveRight\MoveLeft}(n-1,t) - \cos\theta\psi_{\MoveRight\MoveRight}(n-1,t)
\end{array}\right) + \cdots\\
&= \sum_{n=0}^{d-1}\ket{n}\otimes\psi(n,t+1).
\end{align*}

From this, we see that we may write:
\[
\psi(n,t+1) = M_+(\theta)\psi(n+1,t) + M_-(\theta)\psi(n-1,t),
\]
where:
\begin{align*}
M_+(\theta) &= \left(\begin{array}{cccc}
\frac{1}{\sqrt{2}} & \frac{1}{\sqrt{2}} &0&0\\
0&0&\cos\theta & \sin\theta\\
0&0&0&0\\
0&0&0&0
\end{array}\right)\\\\
M_-(\theta) &= \left(\begin{array}{cccc}
0&0&0&0\\
0&0&0&0\\
\frac{1}{\sqrt{2}} & -\frac{1}{\sqrt{2}} &0&0\\
0&0&\sin\theta & -\cos\theta
\end{array}\right)\\\\
\end{align*}

Now, consider the Fourier transform of the function $\psi(n,t)$:
\[
\widetilde{\psi}(k,t) = \sum_{n=0}^{d-1}\psi(n,t)e^{-2\pi ikn/d},
\]
from which we have:
\begin{align*}
\widetilde{\psi}(k,t+1) &= \sum_{n=0}^{d-1}\psi(n,t+1)e^{-2\pi ikn/d}\\
&= \sum_{n=0}^{d-1}\left(M_+(\theta)\psi(n+1,t) + M_-(\theta)\psi(n-1,t)\right)e^{-2\pi ikn/d}\\
&=\left(e^{2\pi ik/d}M_+(\theta) + e^{-2\pi ik/d}M_-(\theta)\right)\sum_{n=0}^{d-1}\psi(n,t) e^{-2\pi ikn/d}\\
&=M_k(\theta)\widetilde{\psi}(k,t),
\end{align*}
where:
\begin{equation}\label{eq:Mk}
M_k(\theta) = \left(\begin{array}{cccc}
\frac{e^{2\pi ik/d}}{\sqrt{2}} & \frac{e^{2\pi ik/d}}{\sqrt{2}} & 0 & 0\\
0 & 0 & e^{2\pi ik/d}\cos\theta & e^{2\pi ik/d}\sin\theta\\
\frac{e^{-2\pi ik/d}}{\sqrt{2}} & -\frac{e^{-2\pi ik/d}}{\sqrt{2}} & 0 & 0\\
0&0&e^{-2\pi ik/d}\sin\theta & -e^{-2\pi ik/d}\cos\theta
\end{array}\right),
\end{equation}
and, so, we have:
\begin{equation}\label{eq:position-Fourier}
\widetilde{\psi}(k,t) = [M_k(\theta)]^t\widetilde{\psi}(k,0).
\end{equation}

Let $\{\ket{\phi_j(k,\theta)}\}_{j=1}^4$ be the (orthonormal) eigenvectors of $M_k(\theta)$ and $\{\lambda_j(k,\theta)\}_{j=1}^4$ be the corresponding eigenvalues.  Then we may write:
\[
M_k(\theta) = \sum_{j=1}^4 \lambda_j(k,\theta)\ket{\phi_j(k,\theta)}\bra{\phi_j(k,\theta)},
\]
and thus:
\begin{equation}\label{eq:Mk-action}
[M_k(\theta)]^t = \sum_{j=1}^4 \lambda^t_j(k,\theta)\ket{\phi_j(k,\theta)}\bra{\phi_j(k,\theta)}.
\end{equation}

Assuming our initial state is of the form $\ket{\psi_0} = \ket{0}\otimes\psi(0,0)$, that is the particle is located only at position $0$ (thus $\psi(n,0) \equiv 0$ for all $n \ne 0$), then we have:
\begin{equation}\label{eq:start-Fourier}
\widetilde{\psi}(k,0) = \sum_{n=0}^{d-1}\psi(n,0)e^{-2\pi ikn/d} = \psi(0,0), \forall k.
\end{equation}

Changing basis, we may write:
\begin{equation}
\widetilde{\psi}(k,0) = \sum_{j=1}^4 \alpha_j(k,\theta)\ket{\phi_j(k,\theta)},
\end{equation}
where:
\begin{equation}\label{eq:alpha}
\alpha_j(k,\theta) = \braket{\phi_j(k,\theta) | \widetilde{\psi}(k,0)} = \braket{\phi_j(k,\theta) | \psi(0,0)}.
\end{equation}

Combining this, with Equations \ref{eq:position-Fourier} and \ref{eq:Mk-action}, yields:
\begin{align*}
\widetilde{\psi}(k,t) &= [M_k(\theta)]^t\sum_{j=1}^4\alpha_j(k,\theta) \ket{\phi_j(k,\theta)}\\
&=\sum_{j=1}^4\alpha_j(k,\theta)\lambda_j^t(k,\theta) \ket{\phi_j(k,\theta)}.
\end{align*}

Inverting the Fourier transform, we find:
\begin{align*}
\psi(n,t) &= \frac{1}{d}\sum_{k=0}^{d-1}\widetilde{\psi}(k,t)e^{2\pi ikn/d}\\
&= \frac{1}{d}\sum_{k=0}^{d-1}e^{2\pi ikn/d} \sum_{j=1}^4\alpha_j(k,\theta)\lambda_j^t(k,\theta) \ket{\phi_j(k,\theta)}.
\end{align*}

Finally, the probability of measuring the particle in position $\ket{n}$ after $t$ time steps (Equation \ref{eq:pr}) is:

\begin{align}
p(n,t,\phi) = \psi^*(n,t)\psi(n,t) = &\frac{1}{d^2}\sum_{k,m=0}^{d-1}\sum_{j,l=1}^4 e^{2\pi i n (m-k)/d} \alpha^*_j(k,\theta)\alpha_l(m,\theta)\notag\\
&\times\left( \lambda_j^*(k, \theta) \lambda_l(m, \theta) \right)^t \braket{\phi_j(k,\theta) | \phi_l(m, \theta)}.\label{eq:pr-closed}
\end{align}

From this equation, we may compute the time averaged distribution (Equation \ref{eq:time-avg}) which, after some algebra and properties of limits with finite sums, we find:

\begin{equation}\label{eq:time-avg-closed-first}
\bar{p}(n,\phi) = \sum_{k,m=0}^{d-1}\sum_{j,l=1}^4 f(n,k,m,j,l,\theta)\cdot g(k,m,j,l,\theta),
\end{equation}
where
\begin{align}
f(n,k,m,j,l,\theta) &= \frac{1}{d^2}\alpha_j^*(k,\theta) \alpha_l(m,\theta) \braket{\phi_j(k,\theta)|\phi_l(m,\theta)} e^{2\pi in(m-k)/d}\\
g(k,m,j,l,\theta) &= \lim_{T\rightarrow \infty} \frac{1}{T} \sum_{t=1}^T \left(\lambda_j^*(k,\theta)\lambda_l(m,\theta)\right)^t.
\end{align}

Using similar arguments as in Ref \cite{QW-cycle-first} (in particular, see the proof of Theorem 3.4 from that source), we know that:
\[
g(k,m,j,l,\theta) = \left\{\begin{array}{ll}
1 & \text{if } \lambda_j(k,\theta) = \lambda_l(m,\theta)\\
0 & \text{otherwise}
\end{array}\right.,
\]
and therefore, Equation \ref{eq:time-avg-closed-first} becomes simply:
\begin{equation}\label{eq:time-avg-closed}
\bar{p}(n) = \sum_{k,m,j,l} f(n,k,m,j,l,\theta),
\end{equation}
where the sum is over all $k,m,j,l$ where $\lambda_j(k,\theta) = \lambda_l(m,\theta)$.

\subsection{Time Averaged Distribution}

We now examine how the memory parameter $\phi$ affects Equation \ref{eq:time-avg-closed} for various cycle sizes $d$.  This amounts to finding the eigensystem of $M_k(\theta)$ (Equation \ref{eq:Mk}) for $k=0,1,\cdots,d-1$, where $\theta = \frac{\pi}{4}(1+\phi)$.  Unfortunately, determining expressions for the eigensystem, for arbitrary $\theta$, proved difficult.  Even in the ``simple'' case of $\phi = 0$, the eigensystem is too complicated to admit any simplification (as far as we could determine).  A similar problem was encountered by the authors of \cite{QW-memory-cycle} (even though there, they had no memory parameter to contend with).  However, following their example, we can easily compute, numerically, the eigensystem of $M_k(\theta)$ and, thus, evaluate $\bar{p}(n,\phi)$.

In this section, we are interested in the following initial states:
\begin{align}
\psi_a &= (1,0,0,0)^T &&= \ket{\MoveLeft\MoveLeft}\\
\psi_b &= \frac{1}{\sqrt{2}}(1,1,0,0)^T &&= \frac{1}{\sqrt{2}}(\ket{\MoveLeft\MoveLeft} + \ket{\MoveLeft\MoveRight})\\
\psi_c &= \frac{1}{2} (1,1,1,1)^T &&= \frac{1}{2}(\ket{\MoveLeft\MoveLeft} + \ket{\MoveLeft\MoveRight} + \ket{\MoveRight\MoveLeft} + \ket{\MoveRight\MoveRight})\\
\psi_d &= \frac{1}{2} (1,1,1,-1)^T &&= \frac{1}{2}(\ket{\MoveLeft\MoveLeft} + \ket{\MoveLeft\MoveRight} + \ket{\MoveRight\MoveLeft} - \ket{\MoveRight\MoveRight})
\end{align}
These are the values we set $\psi(0,0)$ to when computing $\alpha_j(k,\theta)$ (Equation \ref{eq:alpha}) - e.g., $\psi(0,0) = \psi_b$.  We will clearly state which initial state we use for which graphs in the following.

\begin{figure}
  \centering
  \subfloat[$d = 42$; Initial State $\psi_a$]{\includegraphics[width=150pt]{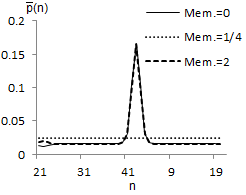}} \qquad
  \subfloat[$d = 11$; Initial State $\psi_b$]{\includegraphics[width=170pt]{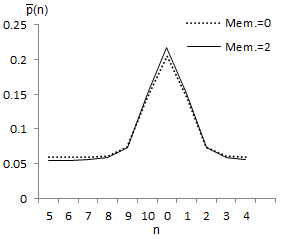}} \qquad\newline\newline
  
  \subfloat[$d = 110$; Initial State $\psi_b$]{\includegraphics[width=150pt]{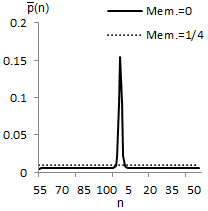}} \qquad
  \subfloat[$d = 11$; Initial State $\psi_c$]{\includegraphics[width=170pt]{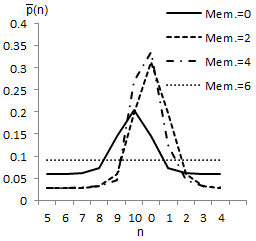}} \qquad
\caption{Graph of the limiting distribution $\bar{p}(n, \phi; \psi_0)$ for various memory parameters $\phi$ (denoted ``Mem.'' in the above figure keys), cycle sizes, and initial states.  In all Figures except (d), setting $\phi = 0$ produced the same distribution as $\phi = 6$; likewise the distribution with $\phi = 2$ was equal to that produced when $\phi = 4$ and so they were not plotted.  In Figure (c) the difference between the distributions for memory parameter $\phi = 0$ and $\phi = 2$ was not noticeable on the graph (though they were different) and so were not plotted; also, this shape was typical for all $d$ we tested with these memory parameters.  Figure (d) produced different distributions for these memory parameters.  All memory parameters not mentioned in either the plot, or the above description, produced the uniform distribution.}\label{fig:init-graphs}
\end{figure}

Our initial results, for various memory parameters, initial states, and cycle sizes, are shown in Figure \ref{fig:init-graphs} and there are many interesting observations to make.  First, there is a relationship between various memory parameters, depending on the initial state:

\begin{theorem}\label{thm:first}
Let $Q$ be the matrix:
\[
Q = \left(\begin{array}{cccc}
1&0&0&0\\
0&1&0&0\\
0&0&1&0\\
0&0&0&-1
\end{array}\right).
\]
For all $d \ge 1$, $t \ge 0$, and initial states $\ket{\psi_0} = \ket{0}_P \otimes \psi(0,0)$, let $\phi' = -(2+\phi)$, then:
\[
p(n,t,\phi; \psi(0,0)) = p(n,t, \phi' ; Q\cdot\psi(0,0)).
\]
\end{theorem}
\begin{proof}
Recall that $\theta = (1+\phi)\frac{\pi}{4}$.  This implies that $\theta' = (1+\phi')\frac{\pi}{4} = -\theta$.

Let $x = \exp(2\pi ik/d)$ and $y = \exp(-2\pi ik/d)$, then we may write $M_k(\theta)$ and $M_k(-\theta)$ (Equation \ref{eq:Mk}) as:

\begin{align*}
M_k(\theta) &= \left(\begin{array}{cccc}
\frac{x}{\sqrt{2}}&\frac{x}{\sqrt{2}}&0&0\\
0&0&x\cos\theta&x\sin\theta\\
\frac{y}{\sqrt{2}}&-\frac{y}{\sqrt{2}}&0&0\\
0&0&y\sin\theta&-y\cos\theta
\end{array}\right)
\end{align*}
\begin{align*}
M_k(-\theta) &= \left(\begin{array}{cccc}
\frac{x}{\sqrt{2}}&\frac{x}{\sqrt{2}}&0&0\\
0&0&x\cos\theta&-x\sin\theta\\
\frac{y}{\sqrt{2}}&-\frac{y}{\sqrt{2}}&0&0\\
0&0&-y\sin\theta&-y\cos\theta
\end{array}\right)
\end{align*}

Let $\lambda$ be an eigenvalue of $M_k(\theta)$ with corresponding eigenvector $\ket{\phi} = (\alpha, \beta, \gamma, \delta)^T$.  Thus:
\begin{equation}\label{eq:ev1-mk0}
M_k(\theta)\ket{\phi} = \left(\begin{array}{c}
\frac{x}{\sqrt{2}}(\alpha+\beta)\\
x(\gamma\cos\theta+\delta\sin\theta)\\
\frac{y}{\sqrt{2}}(\alpha-\beta)\\
y(\gamma\sin\theta-\delta\cos\theta)
\end{array}\right) = \lambda \left(\begin{array}{c}
\alpha\\
\beta\\
\gamma\\
\delta
\end{array}\right).
\end{equation}

We claim $\lambda$ is also an eigenvalue of $M_k(-\theta)$ with corresponding eigenvector $\ket{\phi'} = Q\ket{\phi} = (\alpha, \beta, \gamma, -\delta)^T$.  Indeed:
\begin{align*}
M_k(-\theta)\ket{\phi'} &= \left(\begin{array}{c}
\frac{x}{\sqrt{2}}(\alpha+\beta)\\
x(\gamma\cos\theta+\delta\sin\theta)\\
\frac{y}{\sqrt{2}}(\alpha-\beta)\\
y(-\gamma\sin\theta+\delta\cos\theta)
\end{array}\right)\\\\&= \left(\begin{array}{c}
\frac{x}{\sqrt{2}}(\alpha+\beta)\\
x(\gamma\cos\theta+\delta\sin\theta)\\
\frac{y}{\sqrt{2}}(\alpha-\beta)\\
-y(\gamma\sin\theta-\delta\cos\theta)
\end{array}\right) = \lambda \left(\begin{array}{c}
\alpha\\
\beta\\
\gamma\\
-\delta
\end{array}\right),
\end{align*}
where the last equality follows from Equation \ref{eq:ev1-mk0}.

Thus, let $\ket{\phi_j(k,\theta')} = Q\ket{\phi_j(k,\theta)}$ and $\lambda_j(k,\theta') = \lambda_j(k,\theta)$ for all $k$ and $j$, where $\lambda_j(k,\theta)$ and $\ket{\phi_j(k,\theta)}$ are eigenvalues and (orthonormal) eigenvectors of $M_k(\theta)$.  It is clear that, for every $k,j,m,l$, it holds:
\begin{equation}\label{eq:thm1-eq1}
\braket{\phi_j(k,-\theta)|\phi_l(m,-\theta)} = \braket{\phi_j(k,\theta) | Q^*Q | \phi_l(m,\theta)} = \braket{\phi_j(k,\theta)|\phi_l(m,\theta)}.
\end{equation}
From the above analysis, these $\lambda_j(k,-\theta)$ and $\ket{\phi_j(k,-\theta)}$ are eigenvalues and (orthonormal) eigenvectors of $M_k(-\theta)$.

Let $\psi' = Q\psi(0,0)$ be the starting state for the walk with memory parameter $\phi'$.  Then, from Equation \ref{eq:alpha}, we have:
\begin{align*}
\alpha_j(k,-\theta) &= \braket{\phi_j(k,-\theta)|\psi'} = \braket{\phi_j(k,\theta)|Q^*Q\psi(0,0)}\\
&= \braket{\phi_j(k,\theta)|\psi(0,0)} = \alpha_j(k,\theta).
\end{align*}

From this, along with Equation \ref{eq:thm1-eq1}, we may conclude that $p(n, t, \phi; \psi(0,0))$ $=$ $p(n, t, \phi'; Q\psi(0,0))$.
\end{proof}

Of course the above theorem also implies, trivially, that the limiting distributions are also equal.  That is:
\[
\bar{p}(n, \phi; \psi(0,0)) = \bar{p}(n, \phi'; Q\cdot\psi(0,0)),
\]
where $\phi' = -(2+\phi)$.  Since we typically restrict the memory parameter to be in the interval $[0,8)$ (due to the periodicity of $M_k(\theta)$), we may take the value of $\phi'$ to be modulo 8.  In particular, we have the following identities:
\begin{enumerate}
  \item $\bar{p}(n, 0; \psi(0,0)) = \bar{p}(n, 6; Q\psi(0,0))$.
  \item $\bar{p}(n, 2; \psi(0,0)) = \bar{p}(n, 4; Q\psi(0,0))$.
  \item $\bar{p}(n, 1; \psi(0,0)) = \bar{p}(n, 5; Q\psi(0,0))$.
\end{enumerate}

Other memory parameters, which do not obey the relationship required by Theorem \ref{thm:first} (i.e., those $\phi' \ne -(2+\phi)$), do not necessarily induce equivalent distributions.  For instance, Figure \ref{fig:init-graphs} (b) shows that it is not necessarily the case that $\bar{p}(n, 0; \psi(0,0)) = \bar{p}(n,2; Q\psi(0,0))$.  They produce different distributions, but we can compute the statistical distance between these two distributions to determine how different they are.  That is, we are interested in the value:
\begin{equation}\label{eq:stat-dist-mem}
\frac{1}{2}\sum_{n=0}^{d-1} \left|\bar{p}(n,0; \psi(0,0)) - \bar{p}(n, 2; Q\psi(0,0))\right|,
\end{equation}
for various initial states and cycle sizes $d$.  This distance is shown in Figure \ref{fig:stat-dist}.  When the initial state used was $\psi_a$, we observe that the difference between the distributions induced by memory parameter $\phi = 0$ and that induced by $\phi = 2$ decreases as $d$ increases.  The behavior was also affected by the residue of the cycle size, modulo $4$.  If $d$ is divisible by $4$, there was no difference (Equation \ref{eq:stat-dist-mem} was zero); by Theorem \ref{thm:first}, this would imply memory parameters $\phi = 0,2,4,6$ produce identical distributions for cycles of this size.  If $d = 4r+2$ for some $r$, the difference was greatest.  Second greatest was $d = 4r+1$, though the distance was almost the same as when $d = 4r+3$.  See Figure \ref{fig:stat-dist} (a) and (b).

When using initial state $\psi_c$ and $\psi_d$ we get a very different result.  Here, as the cycle size $d$ increases, the difference between the two distributions does not decrease.  This was for all $d$ we tested (even $d$ divisible by four).  Furthermore, the residue of $d$ modulo $4$ played a less prominent roll in the statistical distance for these two initial states (different residues produced different results, however it was a very slight difference and so was not plotted in these graphs).  See Figure \ref{fig:stat-dist} (c) and (d).

These are all numerical observations: we have no proof of these leaving this as potentially very interesting future work.

\begin{figure}
  \centering
  \subfloat[Initial State $\psi_a$]{\includegraphics[width=160pt]{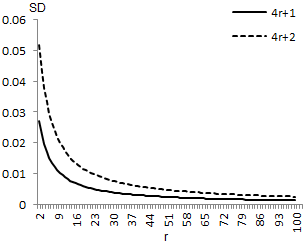}} \qquad
  \subfloat[Initial State $\psi_a$]{\includegraphics[width=160pt]{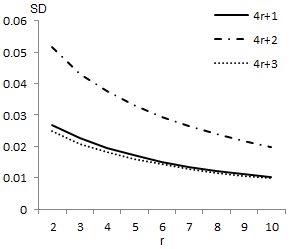}} \qquad\newline\newline
  
  \subfloat[Initial State $\psi_c$]{\includegraphics[width=150pt]{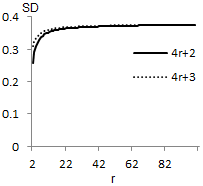}} \qquad
  \subfloat[Initial State $\psi_d$]{\includegraphics[width=170pt]{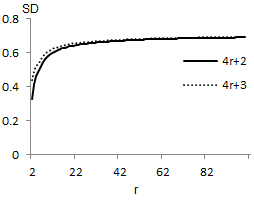}} \qquad
\caption{A graph of the statistical difference between $\bar{p}(n, 0)$ and $\bar{p}(n, 2)$ (see Equation \ref{eq:stat-dist-mem}) for various initial states and cycle sizes of the form $d=4r+j$ with $j \in \{0,1,2,3\}$.  Observe that this value differs based on the residue of the cycle size, modulo $4$, and the initial state used.  In (a) and (b), for $d = 4r$, the difference was zero (i.e., the two memory parameters produced the same distribution) and so was not plotted.  The difference between cycle sizes of the form $d = 4r+1$ and $4r+3$ is only noticeable for small $r$ - thus these are plotted in (b) but not (a).  However, for initial states $\psi_c$ and $\psi_d$, the difference between cycle sizes of the form $4r$ and $4r+2$ was not noticeable - similarly for odd cycle sizes - and so they are not plotted above (there was a difference, however).}\label{fig:stat-dist}
\end{figure}

Our next observation, and perhaps most interesting, is that for any memory parameter $\phi \ne 0, 2, 4,$ or $6$, and for any initial state $\psi(0,0)$, we observed that the time averaged distribution is uniform if the cycle size $d$ is not divisible by 12.  If the cycle size is divisible by $12$, then all memory parameters $\phi \ne 0, 1, 2, 4, 5$, or $6$ produced a uniform distribution (see Figure \ref{fig:graphs12} for the case when $d = 24$).  These observations we were unable to prove, however we did verify it for all cycle sizes $d \le 300$ and initial states $\psi(0,0) = \psi_a, \psi_b, \psi_c$, and $\psi_d$.  For this verification, we computed, for each of the four initial states, Equation \ref{eq:time-avg-closed} numerically for all $d \in \{2, 3, \cdots, 300\}$ and $\phi = m/10$ for $m \in \{0, 1, \cdots, 79\}$.

To ensure that we were not merely running into numerical approximation errors, we also simulated the walk for various $d$ and memory parameters, estimating $\bar{p}(n)$ for $T = 9\times 10^{8}$.  That is, we did not use Equation \ref{eq:time-avg-closed}, but instead applied $U$ (Equation \ref{eq:evolve}) sequentially, calculating $p(n,t, \phi)$ directly for each $t \le T$ and averaging these, thus providing an estimate of $\bar{p}(n)$.  This simulation confirmed our results.  We did notice, however, that the parameter did have an effect on the convergence time, or \emph{mixing time} \cite{QW-cycle-first}.  See Figure \ref{fig:mixingtime}.

\begin{figure}
  \centering
  \subfloat[$d = 24$; Initial State $\psi_a$]{\includegraphics[width=150pt]{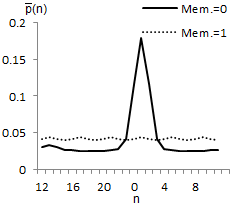}} \qquad
  \subfloat[$d = 24$; Initial State $\psi_c$]{\includegraphics[width=170pt]{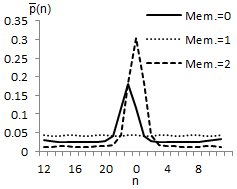}} \qquad\newline\newline
\caption{Time averaged distribution for $d=24$.  In (a), initial state $\psi_a$ was used - here memory parameters $\phi = 0,2,4$ and $6$ produced identical distributions; parameter $\phi = 1, 5$ produced identical distributions which were ``nearly'' but not quite uniform; all other parameters produced the uniform distribution.
Figure (b) was also a cycle of size $24$ however now with initial state $\psi_c$.  Here all integer memory parameters, except for $\phi = 3,7$ produced unique, non-uniform, distributions (only three were plotted above).  All other parameters (including $\phi = 3,7$) produced the uniform distribution.  Similar behavior occurred for all cycle sizes divisible by $12$.  When $d$ is not divisible by $12$, all memory parameters $\phi \ne 0,2,4,6$ seem to produce the uniform distribution.}\label{fig:graphs12}
\end{figure}

\begin{figure}
  \centering
  \subfloat[$d = 11$; Initial State $\psi_b$]{\includegraphics[width=150pt]{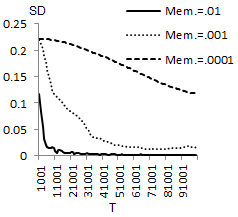}} \qquad
  \subfloat[$d = 11$; Initial State $\psi_b$]{\includegraphics[width=170pt]{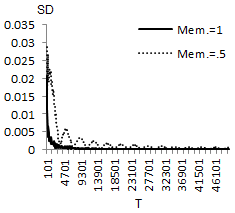}} \qquad\newline\newline
\caption{Let $SD = SD(T) = \frac{1}{2}\sum_{n=0}^{d-1}\left|\frac{1}{T}\sum_{t=0}^{T-1}p(n,t) - \frac{1}{d}\right|$ be the statistical distance between the time averaged distribution for finite $T$ and the uniform distribution.  For the memory parameters above, we've observed that, as $T \rightarrow \infty$, the time averaged distribution converges to uniform.  These graphs demonstrate how the memory parameter affects the rate of convergence.  All graphs use $d = 11$ and initial state $\psi_b$.}\label{fig:mixingtime}
\end{figure}

Finally, Figure \ref{fig:init-graphs} (d) betrays a third possible behavior of this walk model: for initial state $\psi_c$ and memory parameter $\phi = 6$, the limiting distribution appears to be uniform for all odd cycle sizes.  We verified this for all odd $d \le 400$.  As a consequence of Theorem \ref{thm:first}, a similar statement may be made for initial state $\psi_d$ and memory parameter $\phi = 0$.  Note, however, that this did not hold for even cycles (for either $\psi_c$ or $\psi_d$).

Of course, as mentioned, the value of $\phi$ would probably affect the \emph{mixing time} of the distribution - a quantity which depends on the difference between non-equal eigenvalues: the closer the eigenvalues are to one another, the longer it takes for the distribution to converge to its limiting distribution \cite{QW-cycle-first}.  Analyzing this further would be interesting; however, it seems that, in order to progress further in this direction, a simplified expression for the eigenvalues would be required.

In summary, we observe that the behavior of the recycled walk on the $d$-cycle depends on the following:
\begin{enumerate}
  \item The memory parameter $\phi \in [0,8)$; however, we proved that for certain initial conditions, certain memory values induce the same distribution.  Also, we observed (without proof) that only $\phi = 0,1,2,4,5,6$ potentially produced non-uniform behavior.  If this is true, then, combining this with Theorem \ref{thm:first}, this tells us we need only focus on memory parameters $\phi = 0, 1, 2$.
  \item The size of the cycle, modulo $4$ affects the distribution.
  \item If the size of the cycle is divisible by $12$, this also appears to affect the distribution.  In particular, if the cycle size is not divisible by twelve, memory parameters $\phi = 1, 5$ seem to produce the uniform distribution (this is not proven).  If true, then, combining with Theorem \ref{thm:first}, only memory parameters $\phi = 0, 2$ require further investigation (when $d$ is not divisible by $12$).
  \item The initial state also affects the distribution.
\end{enumerate}
We stress that, with the exception of Theorem \ref{thm:first}, these are observations which we were only able to verify numerically.

\subsection{Comparison with Quantum Walk with Memory}

Though the memory model of \cite{QW-memory-cycle}, using a Hadamard coin operator, and the recycled coin model \cite{QW-recy-coin}, using the unbalanced coin from Equation \ref{eq:coin-flip}, are seemingly quite different, we encountered an interesting relationship between the two while performing our analysis.  To distinguish between the memory model and the recycled coin model, when referring to the former, we shall always write any mathematical definition with a subscript $\mathcal{M}$.  For example, the probability of the particle being measured in position $\ket{n}$ after time $t$ is given by $p_\mathcal{M}(n,t)$; the state of the quantum walk in the memory model is $\psi_{\mathcal{M}}(n,t)$; and so on.

With this notation, we observed, and now prove, an interesting relationship between the two models:

\begin{theorem}\label{thm:second}
Let $d \ge 1$, $t \ge 0$, and $\psi(0,0)$ be an initial state for the recycled coin walk.  Let $P$ be the following permutation matrix:
\begin{equation}\label{eq:perm}
P = \left(\begin{array}{cccc}
1&0&0&0\\
0&0&1&0\\
0&0&0&1\\
0&1&0&0\end{array}\right),
\end{equation}
then it holds that $p_\mathcal{M}(n,t;P^*\psi(0,0)) = p(n,t, 2; \psi(0,0))$ for every $n$.  Of course this implies $\bar{p}_\mathcal{M}(n;P^*\psi(0,0)) = \bar{p}(n,2;\psi(0,0))$ trivially.
\end{theorem}
\begin{proof}
In Ref \cite{QW-memory-cycle}, the authors, using Fourier analysis, showed:
\begin{align}
p_{\mathcal{M}}(n,t) = &\frac{1}{d^2}\sum_{k,m=0}^{d-1}\sum_{j,l=1}^4e^{2\pi in(m-k)/d}\alpha_{j,\mathcal{M}}^*(k)\alpha_{l,\mathcal{M}}(m)\label{eq:pr-memory}\\
&\times \braket{\phi_{j,\mathcal{M}}(k)|\phi_{l,\mathcal{M}}(m)}\left(\lambda_{j,\mathcal{M}}^*(k)\lambda_{l,\mathcal{M}}(m)\right)^t,\notag
\end{align}
where $\{\ket{\phi_{j,\mathcal{M}}(k)}\}_{j=1}^4$ and $\{\lambda_{j,\mathcal{M}}(k)\}_{j=1}^4$ are the eigenvectors and corresponding eigenvalues of the matrix:
\[
N_{k} = \frac{1}{\sqrt{2}}\left(\begin{array}{cccc}
e^{2\pi ik/d} &0 & e^{2\pi ik/d} & 0\\
0 & e^{-2\pi ik/d} & 0 & e^{-2\pi ik/d}\\
0 & e^{2\pi ik/d} & 0 & -e^{2\pi ik/d}\\
e^{-2\pi ik/d} & 0 & -e^{-2\pi ik/d} & 0
\end{array}\right),
\]
and $\alpha_{j,\mathcal{M}}(k) = \braket{\phi_{j,\mathcal{M}}(k) | \psi_{\mathcal{M}}(0,0)}$.

We claim that, if $\phi = 2$ (thus $\theta = 3\pi/4$), then for every $k$ and $j$, it holds that $\lambda_j(k,2) = \lambda_{j,\mathcal{M}}(k)$ and $\ket{\phi_j(k,2)} = P\ket{\phi_{j,\mathcal{M}}(k)}$ where $P$ is the permutation matrix from Equation \ref{eq:perm}, and $\{\ket{\phi_{j,\mathcal{M}}(k)}\}_{j=1}^4$, $\{\lambda_{j,\mathcal{M}}(k)\}_{j=1}^4$ is an eigensystem of $N_k$.

Let $x = \exp(2\pi ik/d)$ and $y = \exp(-2\pi ik/d)$, then we may write $M_k(\theta)$ (Equation \ref{eq:Mk}) as:
\[
M_k(\theta) = \frac{1}{\sqrt{2}}\left(\begin{array}{cccc}
x&x&0&0\\
0&0&-x&x\\
y&-y&0&0\\
0&0&y&y\\
\end{array}\right).
\]

We may also write $N_k$ as:
\[
N_k = \frac{1}{\sqrt{2}}\left(\begin{array}{cccc}
x&0&x&0\\
0&y&0&y\\
0&x&0&-x\\
y&0&-y&0\\
\end{array}\right).
\]

Consider one of the eigenvectors of $N_k$, $\ket{\phi_{j,\mathcal{M}}(k)}$ and write it as $\ket{\phi_{j,\mathcal{M}}(k)} = (\alpha, \beta, \gamma, \delta)^T$.  Since $\lambda_{j,\mathcal{M}}(k)$ is its corresponding eigenvalue, it holds that:
\begin{equation}\label{eq:Nk-eigensystem}
N_k\ket{\phi_{j,\mathcal{M}}(k)} = \frac{1}{\sqrt{2}}\left(\begin{array}{c}
x\alpha + x\gamma\\
y\beta + y\delta\\
x\beta - x\delta\\
y\alpha - y\gamma
\end{array}\right) = \lambda_{j,\mathcal{M}}(k)\left(\begin{array}{c}
\alpha\\
\beta\\
\gamma\\
\delta
\end{array}\right).
\end{equation}

Now let $\ket{\phi_j(k,\theta)} = P\ket{\phi_{j,\mathcal{M}}(k)} = (\alpha, \gamma, \delta, \beta)^T$.  We show this is an eigenvector of $M_k(\theta)$ with eigenvalue $\lambda_{j}(k,\theta) = \lambda_{j,\mathcal{M}}(k)$:
\[
M_k(\theta)\ket{\phi_j(k,\theta)} = \frac{1}{\sqrt{2}}\left(\begin{array}{c}
x\alpha + x\gamma\\
x\beta - x\delta\\
y\alpha - y\gamma\\
y\beta + y\delta
\end{array}\right).
\]
Combining this with Equation \ref{eq:Nk-eigensystem}, we find:
\[
M_k(\theta)\ket{\phi_j(k,\theta)} = \lambda_{j,\mathcal{M}}(k)\left(\begin{array}{c}
\alpha\\
\gamma\\
\delta\\
\beta
\end{array}\right) = \lambda_{j}(k,\theta)\ket{\phi_j(k,\theta)}.
\]
Because $P^*P = I$, it is clear that, for every $k,m,j$ and $l$, it holds:
\begin{align*}
\braket{\phi_j(k,\theta)|\phi_l(m,\theta)} = \braket{\phi_{j,\mathcal{M}}(k)|P^* P|\phi_{l,\mathcal{M}}(m)} = \braket{\phi_{j,\mathcal{M}}(k)|\phi_{l,\mathcal{M}}(m)}.
\end{align*}

In particular, this shows us that the $\{\ket{\phi_j(k,\theta)}\}_{j=1}^4$ are orthonormal eigenvectors of $M_k(\theta)$.

If we let the initial state of the memory walk be $\psi_{\mathcal{M}}(0,0) = P^*\psi(0,0) \Rightarrow \psi(0,0) = P\psi_{\mathcal{M}}(0,0)$, then using Equation \ref{eq:alpha}:
\[
\alpha_j(k,\theta) = \braket{\phi_j(k,\theta) | \psi(0,0)} = \braket{\phi_{j,\mathcal{M}}(k) | P^* P |\psi_{\mathcal{M}}(0,0)} = \alpha_{j,\mathcal{M}}(k).
\]

Combining all of this, it is clear that $p_{\mathcal{M}}(n,t;P^*\psi(0,0))$ (Equation \ref{eq:pr-memory}) is exactly equal to $p(n,t,2;\psi(0,0))$ (Equation \ref{eq:pr}) for every $n$ and $t$.
\end{proof}

For example, if the initial state used, for both the recycled coin model and the memory model, is $\psi(0,0) = \ket{\MoveLeft\MoveLeft}$ these walks will produce exactly the same distribution if memory parameter $\phi = 2$ is used.

Thus, to analyze the quantum walk with memory model on the cycle, one may instead analyze the recycled coin model with parameter $\phi = 2$ and a potentially different initial state.  Or vice versa.  It might be that one model yields more easily to further analysis than the other.

\section{CLOSING REMARKS}

In this paper we analyzed the recycled coin discrete time quantum random walk model on the cycle.  Our analysis shows that this model's behavior depends in very interesting ways not only on the memory parameter, a continuous value in the range $[0,8)$, but also on the initial state, and the size of the cycle - in particular whether the cycle size is divisible by $12$ or not (compare Figure \ref{fig:init-graphs} with Figure \ref{fig:graphs12}); also on the residue of the cycle size modulo four (see Figure \ref{fig:stat-dist} for how this residue affects certain distributions).

Our analysis showed that, despite a seemingly wide range of possible behaviors, due to the additional memory parameter $\phi \in [0,8)$, there seem to be only a few cases of interest: in particular, only integer values of $\phi$ seem to potentially produce non-uniform distributions (non-integer values of $\phi$ producing a uniform time averaged distribution).  Of these integer values, only a few need be considered.  In particular, if the cycle size $d$ is divisible by $12$, then we observed only values $\phi = 0, 1, 2, 4, 5, 6$ produced non-uniform behavior.  If $d$ is not divisible by $12$ then only $\phi = 0, 2, 4, 6$ produced a non-uniform distribution.  This is a conjecture we verified for $d \le 300$ and we leave the proof as an open problem.  However, we note that, due to Theorem \ref{thm:first}, only a few of these cases need to be considered: $\phi = 0, 1, 2$ in the first case, and $\phi = 0, 2$ in the second.

Even if the limiting distribution is uniform for all non-integer values of $\phi$ as mentioned above, this memory parameter most likely influences the mixing time, as we mentioned last section.  Deriving a relationship between these two quantities, the parameter and the mixing time, could be potentially useful, and certainly interesting.

Our efforts to prove these last few points were hampered, however, by our inability to derive a simplified expression for the eigenvalues of $M_k(\theta)$ (Equation \ref{eq:Mk}).  Without these, we are impeded in this direction; perhaps an alternative method of proof may be required.

Finally, we only considered one type of coin flip operator - as mentioned in Ref \cite{QW-recy-coin}, there are a multitude of ways to relate the memory parameter with a coin flip operator.  The one described in Equation \ref{eq:coin-flip} was similar to one first mentioned in \cite{QW-recy-coin} and so seemed a good place to start.  It might be useful to consider others in the future - perhaps results similar to our Theorems \ref{thm:first} and \ref{thm:second} can be discovered linking the distributions, produced by these different coin operators, in some manner.

\end{document}